\documentclass[12,reqno]{amsart}
\usepackage{amsfonts}
\usepackage{amssymb}
\usepackage{a4wide}
\usepackage{pgf,pgfarrows,pgfautomata,pgfheaps,pgfnodes,pgfshade}
\usepackage[colorlinks]{hyperref}
\usepackage{url}

\numberwithin{equation}{section}

\newtheorem{defi}{Definition}[section]
\newtheorem{thm}[defi]{Theorem}

\newtheorem{cor}[defi]{Corollary}
\newtheorem{prop}[defi]{Proposition}

\begin{document}

\title[SDE in Random Population Growth]
{SDE in Random Population Growth}
\author{Raouf Ghomrasni}
\address{Programme in Advanced Mathematics of Finance, School of Computational \& Applied Mathematics, University of the Witwatersrand, Private Bag 3, Wits, 2050 Johannesburg, South Africa.}
\email{Raouf.Ghomrasni@wits.ac.za}
\author{Lisa Bonney}
\address{Programme in Advanced Mathematics of Finance, School of Computational \& Applied Mathematics, University of the Witwatersrand, Private Bag 3, Wits, 2050 Johannesburg, South Africa.} \email{Lisa.Bonney@students.wits.ac.za}

\date{}

\maketitle

\begin{abstract}
\noindent In this paper we extend the recent work of C.A. Braumann \cite{B2007} to the case of stochastic differential equation with random coefficients. Furthermore, the relationship of the It\^o-Stratonovich stochastic calculus to studies of  random population growth is also explained.
\end{abstract}


\section{Introduction}

As mentioned in the paper by Carlos A. Braumann \cite{B2007}, there is the issue of the It\^o-Stratonovich controversy. That is, the issue of which stochastic calculus, It\^o or Stratonovich, to adopt in the SDE models of the population growth. It has been shown that It\^o and Stratonovich calculus give different results and do not yield the same solutions to the SDE models, leading to this controversy on which calculus is more appropriate when modelling population growth in a random environment. Hence, creating an obstacle on the use of this modelling approach. Carlos A. Braumann clears up the confusion concerning this controversy by showing that the apparent difference between the It\^o and Stratonovich calculus is due to the confusion based on the assumption that both It\^o and Stratonovich employ the same type of mean rates, i.e. interpreting the mean rate as an unspecified ``average" per capita growth rate. In fact, It\^o and Stratonovich calculus will yield exactly the same results when coupled with the appropriate mean rate. It is proven that, when using It\^o calculus, $b(N)$ is the arithmetic average growth rate $R_a(x,t)$, and when using Stratonovich calculus, $b(N)$ is the geometric average growth rate $R_g(x,t)$.\\

\section{The Model}

\noindent Let $N = N(t)$ denote the population size (number of individuals, density) at time $t$ of a closed population (no migrations) and assume that the initial population size $N(0) = N_0 >0$ is known. In a randomly varying environment, we shall refer to $\frac{dN}{dt}$ as the {\emph{total population growth rate}} and to the {\emph{per capita growth rate}}, $\frac{1}{N} \frac{dN}{dt}$, simply by the {\emph{growth rate}}. \\

\noindent We can model the dynamics by assuming that the growth rate $\frac{1}{N} \frac{dN}{dt}$ is the sum of an ``average" growth rate $b(N)$ and perturbations caused by random environmental fluctuations. We can approximate these perturbations by a white noise $\sigma \xi (t)$. Where the growth rate $\frac {1}{N} \frac{dN}{dt}$ is for some $N > 0$, some density-dependent function $b(N)$ (which is the growth rate of the population) having a continuous derivative.\\ 

\noindent We can model the dynamics using the stochastic differential equation (SDE):

\begin{equation} 
\label{eq:model1} 
\frac{1}{N} \frac{dN}{dt} = b(N) + \sigma \xi (t), \qquad N(0) = N_0 > 0, 
\end{equation}

\noindent where we write the per capita growth rate $\frac{1}{N} \frac{dN}{dt}$ as an ``average" density-dependent rate $b(N)$ perturbed by a white noise $\sigma \xi (t)$, to take into account the effect of random environmental fluctuations. \\

\noindent The $b(N)$ represents the ``average" growth rate in the population size. The $\sigma \xi (t)$ represents the random and uncertain movement in the population size which can be attributed to random environmental fluctuations that perturb the per capita growth rate. This perturbation is assumed to be a stationary stochastic process and can be reasonably approximated by a white noise process. $\sigma(t) > 0$ is the volatility and can be regarded as adding noise or variability to the fluctuations in the population size, and $\xi(t)$ is a standard Gaussian white noise process. Furthermore, $\xi (t)$ is a generalised derivative of the standard Brownian motion process $W(t)$ and is therefore equal to $\frac{dW(t)}{dt}$.\\

\noindent The SDE given in (\ref{eq:model1}) can be rewritten as the following:

\begin{equation} 
\label{eq:model2} 
\frac{dN}{dt} = b(t,\omega) N(t) + \sigma(t,\omega) N(t) \xi (t).
\end{equation}

\noindent Here, $\omega \in \Omega$ represents the random environmental scenario (event) in the set $\Omega$ of all possible environmental scenarios, on the probability space structure, $(\Omega,\mathcal{F},\mathbb{P})$. By a scenario $\omega$, we mean a specific combination of environmental conditions that a population might be subjected to.

\noindent We can rewrite this differential equation (\ref{eq:model2}) in terms of a Brownian motion process, where we simply substitute $\xi (t) dt$ with $dW_t$, since $\int_0^t \xi (s) ds = W(t) = \int_0^t dW(s)$. Hence, the {\emph{basic model}} of the population growth in a random environment is given by:

\begin{equation}
\label{eq:model}
d N(t) = N(t) \Big[ b(t,\omega)dt + \sigma(t,\omega) dW_t \Big], \qquad N(0)= N_0>0,
\end{equation}

\noindent or in integral form 
\begin{equation}
\label{eq:intmodel}
N(t) = N_0 + \int_0^t N(s)b(s)ds + \int_0^t N(s)\sigma(s) dW_s,
\end{equation}

\noindent where $W$ is a standard one-dimensional Brownian motion. $b(\cdot,\omega)$ and $\sigma(\cdot,\omega)$ are assumed to be adapted and satisfy the integrability condition:

\begin{equation*}
\int_0^T (\vert  b(t) \vert + \sigma(t)^2) dt < \infty
\end{equation*}

\noindent almost surely, for every $T \in (0,\infty)$. This integrability condition ensures that $N(t)<\infty \quad \forall t \in \mathbb{R}^+$.\\

\noindent Since $N(t)$ is continuous, the first integral in Equation (\ref{eq:intmodel}), can be defined as a Riemann integral. However, problems now arise with the definition of the second integral of Equation (\ref{eq:intmodel}), in which $W(t)$ oscillates too rapidly to be defined in the usual Riemann-Stieljes sense (which follows ordinary calculus rules). The second integral in Equation (\ref{eq:intmodel}) contains the Brownian motion component which defines this integral as a stochastic integral. This integral cannot usually be defined as a classical Riemann-Stieltjes integral due to the fact that the limits of Riemann-Stieltjes sums differ according to the choice of intermediate points in the integrand function. As a result, there are many alternative definitions of these stochastic integrals according to the choice of such intermediate points. \\

\noindent The most commonly used integrals in the literature are the {\bf{It\^o}} and {\bf{Stratonovich}} integrals. The It\^o integral has nice probabilistic properties, which includes, apart from being a martingale, the property of zero expectation as well as having a convenient expression for its variance. However, it follows non-ordinary calculus rules. Stratonovich calculus, on the other hand, follows ordinary calculus rules. We will examine the problems that arise in the interpretation of (\ref{eq:model}) when the SDE is taken in the It\^o and Stratonovich sense. \\

\noindent In the next two sections, we will introduce the important concepts of the It\^o and Stratonovich calculus respectively. Thereafter, in Section 5, we will discuss and represent the relationship between It\^o and Stratonovich calculus. In Section 6, the controversy itself will be discussed along with its resolution.

\section{It\^o Calculus}

In It\^o calculus we can express the It\^o SDE in the following form:
\begin{equation*} dX_t = \mu (X_t,t) dt + \sigma (X_t,t) dW_t, \end{equation*}
\noindent or equivalently in integral form as
\begin{equation*} X_t = X_0 + \int_{0}^{t} \mu (X_s,s) ds + \int_{0}^{t} \sigma (X_s,s) dW_s. \end{equation*}

\begin{defi} {\bf{The It\^o Integral}} \\
Suppose that $W(t)$ is a Brownian motion process and that $X(t)$ is a stochastic process. Consider a partition of $[0,T]$, $0=t_0 < t_1 < \cdots < t_n = T$, then the It\^o integral of $X$ w.r.t. $W$ is a random variable
\begin{equation*}  \int_{0}^{T} X_t dW_t := \lim_{n \rightarrow \infty} \sum_{j=0}^{n-1} X_{t_j} ( W(t_{j+1}) - W(t_j) ).
\end{equation*}

\end{defi}

\noindent Notice, in the summation, the function $X$ is defined at the left-hand point, i.e. the value of $X$ at the beginning of each timestep is used, this is of crucial importance.

\begin{thm} {\bf{It\^o's Lemma}}\\
Let $X(t)$ be a generalised Brownian motion process or an It\^o process. That is, let $X(t)$ have the following dynamics
\begin{equation*} \label{Itoprocess} dX(t) = a(X_t,t)dt + b(X_t,t)dW_t, \end{equation*}
where $W_t$ is a Brownian motion process.\\
Let $F(X_t,t)$ be a function with continuous second derivatives, where $F$ and $X$ have a functional dependence. \\
Then $F(X_t,t)$ is also an It\^o process and has the following dynamics

\begin{align*} 
dF(X_t,t) &= \frac{\partial F}{\partial t} (x,t) dt + \frac{\partial F}{\partial x} (x,t) dX(t) + \frac{1}{2} \frac{\partial ^2 F}{\partial x^2} (x,t) (dX(t))^2 \\
&= \left ( \frac{\partial F}{\partial t} (x,t) + a(x,t) \frac{\partial F}{\partial x} (x,t) + \frac{1}{2} b^2(x,t) \frac{\partial^2 F}{\partial x^2} (x,t) \right ) dt + b(x,t) \frac{\partial F}{\partial x} (x,t) dW_t.
\end{align*}
\end{thm}

\noindent Hence, $F$ is also an It\^o process, but with adjusted drift rate given by $\frac{\partial F}{\partial t} (x,t) + a(x,t) \frac{\partial F}{\partial x} (x,t) + \frac{1}{2} b^2(x,t) \frac{\partial^2 F}{\partial x^2} (x,t)$ and a scaled variance, $b^2(x,t) \frac{\partial F}{\partial x} (x,t)$.\\

\noindent We may use It\^o Lemma to solve (\ref{eq:model}), in the form 

\begin{equation*}
\label{eq:model-log}
d \ln N(t) = \gamma(t) dt + \sigma(t) dW_t, \qquad N(0)= N_0>0,
\end{equation*}

\noindent where

\begin{equation}
\label{eq:growth-rate}
\gamma(t):=  b(t)- \frac{1}{2}\sigma^2(t) ,
\end{equation}

\noindent or equivalently:
\begin{equation}
\label{eq:model-sol}
N(t) = N_0 \exp \left \{ \int_0^t \gamma(s)ds + \int_0^t \sigma(s) dW_s \right \}.
\end{equation}

\noindent We note from (\ref{eq:model-sol}) that $N(t)>0$ for all $t>0$ provide that $N_0>0$.

\noindent We shall refer to the quantity of (\ref{eq:growth-rate}) as \emph{the rate of growth} of the population $N$, because of the a.s. relationship
\begin{equation}
\label{eq:model-lim}
\lim_{T \to \infty} \frac{1}{T} \Big( \ln N(T) - \int_0^T \gamma(s)ds\Big) =0, 
\end{equation}

\noindent valid when the variance $a(\cdot)=\sigma^2(\cdot)$ is bounded, uniformly in
$(t, \omega)$; this follows from the strong law of large numbers and from the representation of (local) martingales as time-changed Brownian motions.

\section{Stratonovich Calculus}

In probability theory, the Stratonovich integral is a stochastic integral, the most common alternative to the It\^o integral. The appeal of Stratonovich calculus is that in certain circumstances, integrals in the Stratonovich definition are easier to manipulate. Unlike the It\^o integral counterpart, it is defined such that the chain rule of ordinary calculus holds for the stochastic integrals. Perhaps the most common situation in which these are encountered is as the solution to SDEs. These Stratonovich SDEs are equivalent to It\^o SDEs, apart from the notation $\sigma (X_t,t) \circ dW_t$, where the $``\circ "$ simply indicates that we are working in the Stratonovich sense. Furthermore, it is possible to convert between the two whenever one definition is more convenient for our purposes. \\

\noindent In Stratonovich calculus we can express the Stratonovich SDE in the following form
\begin{equation*} dX_t = \mu (X_t,t) dt + \sigma (X_t,t) \circ dW_t, \end{equation*}
\noindent or equivalently in integral form as
\begin{equation*} X_t = X_0 + \int_{0}^{t} \mu (X_s,s) ds + \int_{0}^{t} \sigma (X_s,s) \circ dW_s. \end{equation*}

\begin{defi} {\bf{The Stratonovich Integral}} \\
It is defined in a similar manner to the Riemann integral, i.e. as a limit of Riemann sums.\\
Suppose that $W(t)$ is a Brownian motion process and that $X(t)$ is a stochastic process. Consider a partition of $[0,T]$, $0=t_0 < t_1 < \cdots < t_n = T$, then the Stratonovich integral of $X$ w.r.t. $W$ is a random variable
\begin{equation*}  \int_{0}^{T} X_t \circ dW_t := \lim_{n \rightarrow \infty} \sum_{j=0}^{n-1} X_{(t_{j+1} + t_{j})/2} ( W(t_{j+1}) - W(t_j) ).
\end{equation*}

\end{defi}

\noindent Here, the function $X$ is evaluated in the middle of each timestep (i.e. choose value of process at midpoint of each subinterval). In the definition of the It\^o integral, the same procedure is used except for choosing the value of the process $X$ at the left-hand point of each subinterval, i.e. $ X_{t_j}$ in place of $X_{(t_{j+1} + t_j)/2}$.\\

\noindent In a similar fashion to It\^o's Lemma, Stratonovich satisfies
\begin{equation*} dF(X_t,t) = \left ( \frac{\partial F}{\partial t} (x,t) + a(x,t) \frac{\partial F}{\partial x} (x,t) \right ) dt + b(x,t) \frac{\partial F}{\partial x} (x,t) dW_t. 
\end{equation*}

\section{The It\^o-Stratonovich Relationship}

\begin{thm} {\bf{Conversion Formula}}\\
Let $X$ be a stochastic process, in particular, an It\^o process satisfying the SDE, $dX_t = a(X_t,t) dt + \sigma (X_t,t) dW_t$. Let $\sigma$ be some function of $X$ and $t$. Furthermore, to make apparent the distinction between the two integrals, we will adopt the subscript ${\boldsymbol{I}}$ to indicate an {\bf{It\^o integral}} and the subscript ${\boldsymbol{S}}$ to indicate a {\bf{Stratonovich integral}}.\\

\noindent Then conversion between the It\^o and Stratonovich integrals may be performed using the formula
\begin{equation} \int_{0}^{T} \sigma_{\boldsymbol{S}} (X_t,t) dW_t = \int_{0}^{T} \sigma_{\boldsymbol{I}} (X_t,t) dW_t + \frac{1}{2} \int_{0}^{T} \frac{\partial \sigma}{\partial x} (x,t) \sigma (X_t,t) dt. \end{equation} 
\end{thm}

\begin{proof}
\begin{align*} \int_{0}^{T} \sigma_{\boldsymbol{S}} (X_t,t) dW_t - &\int_{0}^{T} \sigma_{\boldsymbol{I}} (X_t,t) dW_t = \\
&\lim_{n \rightarrow \infty} \sum_{j=0}^{n-1} \left ( \sigma \left ( \frac{X(t_j) + X(t_{j+1})}{2},t \right ) - \sigma \left ( \frac{X(t_j) + X(t_j)}{2},t \right ) \right ) \left ( W(t_{j+1}) - W(t_j) \right ) 
\end{align*}

\noindent By the Mean Value Theorem, we obtain
\begin{align*} \qquad &= \lim_{n \rightarrow \infty} \sum_{j=0}^{n-1} \frac{\partial \sigma}{\partial x} (X,t) \left ( \frac{X(t_j) + X(t_{j+1})}{2} - X(t_j) \right ) \left (W(t_{j+1}) - W(t_j) \right ) \\
&= \frac{1}{2} \lim_{n \rightarrow \infty} \sum_{j=0}^{n-1} \frac{\partial \sigma}{\partial x} (X,t) (X(t_{j+1}) - X(t_j))(W(t_{j+1}) - W(t_j))  \\
&= \frac{1}{2} \lim_{n \rightarrow \infty} \sum_{j=0}^{n-1} \frac{\partial \sigma}{\partial x} (X,t) \Delta X_t \Delta W_t \\
&= \frac{1}{2} \lim_{n \rightarrow \infty} \sum_{j=0}^{n-1} \frac{\partial \sigma}{\partial x} (X,t) \sigma(X,t) \Delta t \\
&= \frac{1}{2} \lim_{n \rightarrow \infty} \sum_{j=0}^{n-1} \frac{\partial \sigma}{\partial x} (X,t) \sigma(X,t) (t_{j+1} - t_j).
\end{align*}

\noindent Which is the Riemann sum of
\begin{equation*} \frac{1}{2} \int_{0}^{T} \frac{\partial \sigma}{\partial x} (x,t) \sigma (X_t,t) dt. 
\end{equation*}
\noindent Giving us the required result.
\end{proof}

\noindent From Theorem 5.1, it is evident that the It\^o SDE, 
\begin{equation} dX(t) = a(X_t,t)dt + \sigma (X_t,t)dW_t \end{equation} 
is equivalent to the Stratonovich SDE, 
\begin{equation} \label{eq:convert1} dX(t) = \left [ a(X_t,t) - \frac{1}{2} \sigma(X_t,t) \frac{\partial \sigma}{\partial x} (X_t,t) \right ]dt + \sigma(X_t,t) dW_t \end{equation}

\noindent and that the Stratonovich SDE, 
\begin{equation} dX(t) = \alpha (X_t,t)dt + \beta (X_t,t)dW_t \end{equation}
is equivalent to the It\^o SDE
\begin{equation} \label{eq:convert2} dX(t) = \left [ \alpha (X_t,t) + \frac{1}{2} \beta (X_t,t) \frac{\partial \beta}{\partial x} (X_t,t) \right ] dt + \beta (X_t,t) dW_t. \end{equation}

\noindent In the population dynamics context, the Stratonovich SDE, given by

\begin{equation} 
\label{eq:Strat2}
dN(t) = b(t)N(t)dt +  \sigma (t) N(t) dW(t).
\end{equation}

\noindent Using the It\^o-Stratonovich conversion formula (\ref{eq:convert2}), we have $a(N_t,t) = b(t)N(t)$ and $\sigma (N_t,t) = \sigma(t) N(t)$, where $\frac{\partial \sigma}{\partial x} (N_t,t) = \sigma (t)$. Hence, (\ref{eq:Strat2}) is equivalent to the It\^o SDE

\begin{align} 
\label{eq:Ito2}
dN(t) &= \left [b(t)N(t) - \frac{1}{2} \sigma(N_t,t) \sigma(t) \right ]dt +  \sigma (N_t,t) N(t) dW(t)\\
&= \left [b(t)N(t) - \frac{1}{2} \sigma^2 (t) N(t) \right ] dt + \sigma(t) N(t) dW_t.
\end{align}

\noindent In a similar fashion we can obtain the reverse conversion formula.

\section{The Controversy}

\noindent Many qualitative differences have been uncovered between It\^o and Stratonovich calculus. In particular, there are instances in which Stratonovich calculus predicts, for the population, non-extinction and the existence of a stochastic equilibrium, whereas, at the same time, It\^o calculus will predict population extinction. So, it seems, which calculus one uses does have important consequences. This fact has resulted in there being much controversy over which calculus is more appropriate to employ when finding a solution to the SDE.\\

\noindent Considering the dramatic differences in predictions concerning important issues like extinction, which calculus should one trust? This is a major obstacle to the use of these stochastic models.\\

\noindent Braumann (2003) resolved the issue of the It\^o-Stratonovich controversy for the {\bf{density-independent}} growth model, where $b(N) \equiv b$ is identically constant, in a random population environment. Braumann then extended these results in a random environment for the general {\bf{density-dependent}} population growth model. It is revealed that the possible reason for this controversy, is the subtle fact that the same per capita ``average" growth rate, `$b$' is used in both the It\^o and Stratonovich calculus. Therefore, the issue here regards the meaning and interpretation of this average, since it is not elucidated what type of ``average" is being referred to. Furthermore, it is inherently assumed that both the It\^o and Stratonovich calculus make use of the {\emph{same}} average, this is of course an incorrect assumption!\\

\noindent Hence, this issue of the ``average" needs to be addressed and clarification needs to be made of what type of average each method uses. In fact, it is found that the interpretation of $b$ is different when considering population dynamics. When one decides to use It\^o calculus in obtaining a solution to the SDE, $b$ is interpreted as the {\emph{arithmetic average}} growth rate. However, if ones chooses to implement Stratonovich calculus, $b$ is interpreted as the {\emph{geometric average}} growth rate. The differences between these two types of averages results in the dramatic differences between the It\^o and Stratonovich calculus to disappear, yielding exactly the same solutions in both instances.\\

Thus the differences are merely due to the absence of clarification of the meaning of $b$. So, all that is required is to match the appropriate average with the correct type of calculus, and exact same results will be obtained, putting to rest the It\^o-Stratonovich controversy.

\subsection{Types of Averages}

We will denote by $\mathbb{E}_{t,x}[\cdot] := \mathbb{E} [\cdot \vert N(t) = x]$, as the expectation  conditioned on the knowledge that at time $t$ the population size $N(t)$ is $x$.

\subsubsection{The Arithmetic Average}

\noindent The first type of average we have already mentioned is the arithmetic average. It is simply given by the usual expected value, conditioned on the knowledge that $N(t) = x$, and given by $\frac{1}{n} \sum_{i=1}^{n} x_i = \mathbb{E} [X]$. 

\begin{defi}
We define the arithmetic average growth rate at time $t$, when the population size at time $t$ is $x$, as
\begin{equation}
\label{eq:arithmetic} R_a(x,t):= \frac{1}{x} \lim_{\Delta t \downarrow 0} \frac{\mathbb{E}_{t,x}[N(t+\Delta t)]-x}{\Delta t}.
\end{equation}
\end{defi}

\begin{thm}
Let $R_a(x,t)$ be the arithmetic average growth rate as defined above, we have equivalently
\begin{equation}
\label{eq:arith-av}
R_a(x,t)=\mathbb{E}_{t,x}[b(t)].
\end{equation}
\end{thm}

\begin{proof}
We have
\begin{equation*}
\label{eq:model-sol-2}
N(t+\Delta t) = N(t) \exp \left \{ \int_t^{t+\Delta t} \gamma(s)ds + \int_t^{t+\Delta t} \sigma(s) dW_s \right \}.
\end{equation*}
It follows
\begin{align*}
R_a(x,t)&= \frac{1}{x} \lim_{\Delta t \downarrow 0}\frac{\mathbb{E}_{t,x}[\ N(t+\Delta t)]-x}{\Delta t}\\
&= \frac{1}{x}\lim_{\Delta t \downarrow 0} \frac{\mathbb{E}_{t,x}[x\exp\{\int_t^{t+\Delta t} \gamma(s)ds + \int_t^{t+\Delta t} \sigma(s) dW_s\}-x]}{\Delta t} \\
&=\lim_{\Delta t \downarrow 0} \frac{\mathbb{E}_{t,x}[(\exp\{\int_t^{t+\Delta t} b(s)ds\}-1 )\times \exp\{ \int_t^{t+\Delta t} \sigma(s) dW_s- \frac{1}{2}\int_t^{t+\Delta t} \sigma^2(s)ds \}]}{\Delta t} \\
&=\mathbb{E}_{t,x}[b(t)].
\end{align*}
\end{proof}

\subsubsection{The Geometric Average}

\begin{defi}
The geometric mean of a positive random variable $X$ is defined as $e^{\mathbb{E}[\ln(X)]}$
\end{defi}

\noindent Since 
\begin{align*} \mu_{geom} &= \left ( \prod_{i=1}^{n} x_i \right )^{\frac{1}{n}}\\
                               &= e^{\ln [ ( \prod x_i)^{\frac{1}{n}}]}\\
                               &= e^{\frac{1}{n} \ln \left ( \prod x_i \right )}\\
                               &= e^{\frac{1}{n} \sum \ln(x_i)}\\
                               &= e^{\mathbb{E} [\ln (X)]}
\end{align*}

\begin{defi}
We define the geometric average growth rate at time $t$, when the population size at time $t$ is $x$, as
\begin{equation}
\label{eq:geometric1} R_g(x,t):= \frac{1}{x} \lim_{\Delta t \downarrow 0}\frac{\exp(\mathbb{E}_{t,x}[\ln N(t+\Delta t)])-x}{\Delta t}.
\end{equation}
\end{defi}

\noindent Hence, the geometric average is obtained by transforming the quantities to be averaged to log scale then taking an ordinary arithmetic average and then revert to the initial scale by inverting the algorithm.

\begin{prop}
Let $R_g(x,t)$ be the geometric average growth rate as defined above, we have equivalently
\begin{equation}
\label{eq:geometric2} R_g(x,t)= \lim_{\Delta t \downarrow 0}\frac{\mathbb{E}_{t,x}[\ln N(t+\Delta t)]-\ln x}{\Delta t}.
\end{equation}
\end{prop}
\begin{proof}
This follows from the fact that when $z \to 0$ we have $(e^z-1)/z \to 1$. We apply this result to $z:=\mathbb{E}_{t,x}[\ln (N(t+\Delta t)/x)]$ which tends to $0$ when $\Delta t \downarrow 0$ as follows:
\begin{align}
\lim_{\Delta t \downarrow 0} \frac{\mathbb{E}_{t,x}[\ln (N(t+\Delta t)/x)]}{\Delta t}&= \lim_{\Delta t \downarrow 0} \frac{z}{\Delta t} \nonumber\\
&=\lim_{\Delta t \downarrow 0} \frac{(e^z-1)}{z} \times \frac{z}{\Delta t}\nonumber \\
&=\lim_{\Delta t \downarrow 0} \frac{(e^z-1)}{\Delta t} \nonumber \\
&=\frac{1}{x}\lim_{\Delta t \downarrow 0} \frac{(xe^z-x)}{\Delta t}\nonumber \\
&=R_g(x,t).\nonumber
\end{align}
\end{proof}

\begin{thm}
Let $R_g(x,t)$ be the geometric average growth rate as defined above, we have equivalently
\begin{equation}
\label{eq:geo-av}
R_g(x,t) =\mathbb{E}_{t,x}[\gamma(t)].
\end{equation}

\noindent where $\gamma(t)$ is as defined in Equation (\ref{eq:growth-rate}).

\end{thm}
\begin{proof}
We have
\begin{equation*}
\label{eq:model-sol-2bis}
N(t+\Delta t) = N(t) \exp \left \{ \int_t^{t+\Delta t} \gamma(s)ds + \int_t^{t+\Delta t} \sigma(s) dW_s \right \}.
\end{equation*}
By Proposition 6.5 above we have
\begin{align*}
R_g(x,t)&= \lim_{\Delta t \downarrow 0}\frac{\mathbb{E}_{t,x}[\ln N(t+\Delta t)]-\ln x}{\Delta t}\\
&= \lim_{\Delta t \downarrow 0} \frac{\mathbb{E}_{t,x}[\int_t^{t+\Delta t} \gamma(s)ds + \int_t^{t+\Delta t} \sigma(s) dW_s]}{\Delta t} \\
&=\lim_{\Delta t \downarrow 0} \mathbb{E}_{t,x} \left [\frac{1}{\Delta t}\int_t^{t+\Delta t} \gamma(s)ds \right ]\\
&=\mathbb{E}_{t,x}[\gamma(t)].
\end{align*}
\end{proof}

\begin{cor}
When $\gamma(t,\omega)= b(N_t)-\frac{1}{2}\sigma^2$.\\
Let $R_g(x)$ be the geometric average growth rate as defined above, we have equivalently
$$
R_g(x,t)=b(x)-\frac{1}{2}\sigma^2.
$$
\end{cor}

\noindent It seems as though the arithmetic and geometric averages are equivalent, where the geometric average substitutes the process $N(t)$ by the process $\ln N(t)$. These definitions give an indication on how these two rates can be estimated from observed data. For example, to determine an estimate of $R_a(x,t)$, we look at all the instances $t$ for which $N(t)$ is close to $x$ and then we take the average of those $N(t + \Delta t)$ as an approximation of $\mathbb{E}_{t,x} [N(t + \Delta t)]$.

\subsubsection{Other Types of Averages}

\noindent These two averages are not the only averages we can consider, there are many other types of averages that are possible, so when we refer to an ``average" it is of extreme importance to specify which particular average is being referred to. Some other possible averages are the: harmonic, median and quadratic averages.\\

\noindent Median \qquad $\hat{\mu} = med(x_1,\ldots,x_n)$ \\
\noindent Quadratic \quad $\hat{\mu} = \left ( \frac{1}{n} \sum_{i=1}^{n} x_i^2 \right )^\frac{1}{2} = (\mathbb{E} [X^2])^\frac{1}{2}$

\subsection{Density-Independent Growth}

\noindent Recall that there are two main ways in which to interpret SDEs: It\^o calculus and Stratonovich calculus. They usually lead to different solutions. Let us first consider the {\bf{density-independent}} growth rate model in a random environment, this model corresponds to a constant ``average" growth rate, $b(N) \equiv b$,
\begin{equation*} dN(t) = b N dt + \sigma (t) N(t) dW(t). \end{equation*}

\noindent To make the distinction between the two approaches more apparent, we will use the notation $b_{\boldsymbol{I}}$, to denote the ``average" growth rate under {\bf{It\^o calculus}}, and $b_{\boldsymbol{S}}$ to denote the ``average" growth rate under {\bf{Stratonovich calculus}}. 

\subsubsection{It\^o Model}

\noindent Let us consider the density-independent It\^o calculus model:

\begin{equation} \label{eq:mod1} dN(t) = b_{\boldsymbol{I}} N dt + \sigma (t) N(t) dW(t). \end{equation}

\noindent It is convenient to work in the logarithmic scale by making the change of variables, $Y(t) = \ln N(t)$, $y = \ln x$.\\ 

\noindent Applying It\^o's Lemma to Equation (\ref{eq:mod1}), where $F( N(t)) = \ln N(t)$, we obtain
\begin{align} d \ln N(t) &= \frac{\partial F (N)}{\partial x} dN(t) + \frac{1}{2} \frac{\partial ^2 F(N)}{\partial x^2} (d N(t))^2 \\
&= \frac{1}{N(t)} dN(t) + \frac{1}{2} \left (\frac{-1}{N(t)^2} \right ) (dN(t))^2\\
\label{eq:2mod} &= \left ( b_{\boldsymbol{I}} - \frac{1}{2} \sigma^2 (t) \right )dt + \sigma(t) dW(t). 
\end{align}

\noindent Therefore, $Y(t) = \ln N(t)$, satisfies the SDE in Equation (\ref{eq:mod1}).\\

\noindent This can alternatively be expressed in the equivalent integral form as
\begin{align} Y(t) &= Y_0 + \int_{0}^{t} \left ( b_{\boldsymbol{I}} - \frac{1}{2} \sigma^2 \right ) ds + \int_{0}^{t} \sigma dW_s \\
\label{eq:3mod} &= Y_0 + \left (b_{\boldsymbol{I}} - \frac{1}{2} \sigma^2 \right ) t + \sigma W_t. 
\end{align}

\noindent From Equation (\ref{eq:3mod}), we can conclude that $Y(t) \sim N(Y_0 + (b_{\boldsymbol{I}} - \frac{1}{2} \sigma^2 ) t , \sigma^2 t)$. From this we obtain a solution to Equation (\ref{eq:mod1}), which is represented as 

\begin{equation*} N(t) = N_0 \exp \left [ \left ( b_{\boldsymbol{I}} - \frac{1}{2} \sigma^2 \right ) t + \sigma W_t \right ]. \end{equation*}

\noindent Hence, $N(t)$ has a {\emph{lognormal distribution}} with expected value the knowledge that $ \mathbb{E} [e^{\sigma W_t}] = e^{\frac{1}{2} \sigma^2 t}$, given by:
\begin{equation*} \mathbb{E} [N(t)] = N_0 \exp [b_{\boldsymbol{I}} t]. \end{equation*} 

\noindent From Equation (\ref{eq:3mod}), one obtains the asymptotic result $Y(t) \sim (b_{\boldsymbol{I}} - \frac{1}{2} \sigma^2 ) t$ as $t \rightarrow +\infty$. Therefore, as $t \rightarrow \infty$, $N(t) \rightarrow \infty$ or $N(t) \rightarrow 0$ according to whether the ``average" growth rate $b_{\boldsymbol{I}}$ is larger than $\frac{\sigma^2}{2}$ or smaller than $\frac{\sigma^2}{2}$.

\subsubsection{Stratonovich Model}

Let us consider the density-independent Stratonovich calculus model
\begin{equation} \label{eq:1stratmod} dN(t) = b_{\boldsymbol{S}} N dt + \sigma (t) N(t) dW(t). 
\end{equation}

\noindent Since Stratonovich calculus obeys the ordinary calculus rules, we have by the ordinary chain rule of differentiation
\begin{align} dY(t) = d \ln N(t) &= \frac{d \ln N(t)}{d N(t)} dN(t)\\
&= \frac{1}{N(t)} dN(t) \\
\label{eq:2stratmod} &= b_{\boldsymbol{S}} dt + \sigma dW(t).
\end{align} 

\noindent Therefore, $Y(t) = \ln N(t)$ satisfies the SDE given in Equation (\ref{eq:2stratmod}). \\

\noindent This can alternatively be expressed in the equivalent integral form as
\begin{equation} \label{eq:integral} Y(t) = Y_0 + \int_{0}^{t} b_{\boldsymbol{S}} ds + \int_{0}^{t} \sigma dW_s. \end{equation}

\noindent From which one immediately obtains the solution
\begin{equation} \label{eq:3stratmod} Y(t) = Y_0 + b_{\boldsymbol{S}} t + \sigma W_t. \end{equation}

\noindent Since the integrand is constant, in this density-independent case, the It\^o and Stratonovich integrals coincide. For both approaches, we have $\int_{0}^{t} \sigma W_t = \sigma (W_t - W_0) = \sigma W_t$ since $W_0 = 0$.

\noindent Since $W(t) \sim N(0,t)$ (i.e. normally distributed with mean zero and variance $t$), we conclude that $ Y(t) \sim N(Y_0 + b_{\boldsymbol{S}}, \sigma^2 t)$. From this we obtain a solution to Equation (\ref{eq:1stratmod}), which is given by

\begin{equation*} N(t) = N_0 \exp \left [ b_{\boldsymbol{S}} t + \sigma W_t \right ]. \end{equation*}

\noindent Hence, $N(t)$ has a {\emph{lognormal distribution}} with expected value given by:

\begin{equation*} \mathbb{E} [N(t)] = N_0 \exp \left [ \left ( b_{\boldsymbol{S}} + \frac{1}{2} \sigma^2 \right ) t \right ]. \end{equation*}

\noindent From Equation (\ref{eq:3stratmod}), one obtains the asymptotic result $ Y(t) \sim b_{\boldsymbol{S}} t$ as $t \rightarrow +\infty$ (since $\frac{W_t}{t} \rightarrow 0$ as $t \rightarrow \infty$). Therefore, as $t \rightarrow \infty$, $N(t) \rightarrow \infty$ (i.e. growth without bound) or $N(t) \rightarrow 0$ (extinction) according to whether the ``average" growth rate $b_{\boldsymbol{S}}$ is positive or negative.

\subsubsection{Conclusion}

The long-term behaviour of $N(t)$ for both interpretations of the SDE can be further analysed by examining the trajectory of $N(t)$ in probability. Since, $\frac{W_t}{t} \rightarrow 0$ a.s. when $t \rightarrow \infty$, one easily notices that under Stratonovich calculus, $N(t) \rightarrow \infty$ when $b_{\boldsymbol{S}} > 0$ (probability of extinction is zero and there is a stochastic equilibrium) and $N(t) \rightarrow 0$ when $b_{\boldsymbol{S}} < 0$ (i.e. population extinction occurs with probability one). \\

\noindent In a similar fashion, under It\^o calculus, $N(t) \rightarrow \infty$ when $b_{\boldsymbol{I}} > \frac{\sigma^2}{2}$ and $N(t) \rightarrow 0$ when $b_{\boldsymbol{I}} < \frac{\sigma^2}{2}$. \\ 

\noindent The differences between the It\^o and Stratonovich approaches are now apparent. The behaviour appears to be different from the Stratonovich calculus. Hence, if one employs It\^o instead of Stratonovich, the conditions for non-extinction and existence of a stochastic equilibrium are qualitatively different. This illustrates the consequences of the two approaches for the population behaviour. Using Stratonovich calculus, extinction would occur a.s. if the ``average" growth rate $b < 0$, but with It\^o calculus, one can have extinction a.s. even for positive values of the ``average" growth rate $b$ if $b < \frac{\sigma^2}{2}$. \\

\noindent Furthermore, the It\^o calculus obtains different results compared to the deterministic model, this makes It\^o calculus quite popular in modelling, hence, avoiding the issue of ignoring random environmental fluctuations. \\

\noindent The approach taken here works for all density-dependent models and completely and exactly elucidates the difference between the two interpretations. It also exactly solves the problem of which calculus to use and how to use it.

\begin{align*} {\textnormal{It\^o SDE:}} \qquad \qquad N(t) &= b_{\boldsymbol{I}} N dt + \sigma N dW_t \\
d\ln N(t) &= \left ( b_{\boldsymbol{I}} - \frac{1}{2} \sigma^2 \right ) dt + \sigma dW_t \\
{\textnormal{Stratonovich SDE:}} \qquad \qquad dN(t) &= b_{\boldsymbol{S}} N dt + \sigma N dW_t \\
d \ln N(t) &= b_{\boldsymbol{S}} dt + \sigma dW_t 
\end{align*}
 
\subsubsection{Resolution of which Average to use}

\noindent {\bf{6.2.4.1 It\^o}} \\

\noindent Let us compute the two averages for the It\^o SDE (\ref{eq:mod1}), we obtain from (\ref{eq:3stratmod}),

\begin{equation*} Y(t+\Delta t) = \ln x + \left ( b_{\boldsymbol{I}} - \frac{1}{2} \sigma^2 \right ) \Delta t + \sigma ( W(t+\Delta t) - W(t)).
\end{equation*}

\noindent Therefore, $Y(t+\Delta t)$ is normally distributed with mean $\ln x + (b_{\boldsymbol{I}} - \frac{\sigma^2}{2})\Delta t$ and variance $\sigma^2 \Delta t$. The conditional expectation is
\begin{align*} \mathbb{E}_{t,x} [Y(t+\Delta t)] &= \mathbb{E}_{t,x} [\ln N(t+ \Delta t)]\\
                                               &= \ln x + \left ( b_{\boldsymbol{I}} - \frac{\sigma^2}{2} \right )\Delta t.
\end{align*}

\noindent Replacing into equation (\ref{eq:geometric1}), we obtain

\begin{align*} R_g(x,t) &= \frac{1}{x} \lim_{\Delta t \downarrow 0} \frac{\exp (\ln x + (b_{\boldsymbol{I}} - \frac{\sigma^2}{2})\Delta t) - x}{\Delta t}\\
&= \lim_{\Delta t \downarrow 0} \frac{\exp [(b_{\boldsymbol{I}} - \frac{\sigma^2}{2})\Delta t] - 1}{\Delta t}\\
&\equiv b_{\boldsymbol{I}} - \frac{1}{2} \sigma^2 \equiv b_{\boldsymbol{S}}.
\end{align*}

\noindent We also notice that $N(t+\Delta t) = \exp(Y(t+\Delta t))$ is lognormal with parameters $\ln x + (b_{\boldsymbol{I}} - \frac{\sigma^2}{2})\Delta t$ and $\sigma^2 \Delta t$, and so its conditional expectation is

\begin{align*} \mathbb{E}_{t,x} [N(t+\Delta t)] &= \mathbb{E}_{t,x} [\exp(Y(t+\Delta t)]\\
&= \mathbb{E}_{t,x} \left [\exp \left (\ln x + \left (b_{\boldsymbol{I}} - \frac{\sigma^2}{2} \right ) \Delta t + \sigma (W(t+\Delta t) - W(t)) \right ) \right ]\\
&=\exp \left (\ln x + \left ( b_{\boldsymbol{I}} - \frac{\sigma^2}{2} \right ) \Delta t + \frac{1}{2} \sigma^2 \Delta t \right )\\
&= x\exp(b_{\boldsymbol{I}} \Delta t).
\end{align*}
 
\noindent Replacing into equation (\ref{eq:arithmetic}), we obtain

\begin{equation*} R_a(x,t) = \frac{1}{x} \lim_{\Delta t \downarrow 0} \frac{x\exp(b_{\boldsymbol{I}} \Delta t) - x}{\Delta t} \equiv b_{\boldsymbol{I}}.
\end{equation*}

\noindent The conclusion is that
\begin{align*} R_g(x,t) &\equiv b_{\boldsymbol{I}} - \frac{\sigma^2}{2}, \\
               R_a(x,t) &\equiv b_{\boldsymbol{I}}.
\end{align*}
               
\noindent Hence, when using It\^o calculus, the ``average" growth rate $b_{\boldsymbol{I}}$ is specified as the {\bf{arithmetic average}} growth rate.\\

\noindent {\bf{6.2.4.2 Stratonovich}} \\

\noindent Let us compute these two averages for the Stratonovich SDE model (\ref{eq:1stratmod}). Since $N(t) = x$, we obtain from (\ref{eq:3stratmod}),

\begin{equation*} 
\label{eq1} Y(t + \Delta t) = \ln x + b_{\boldsymbol{S}} \Delta t + \sigma ( W(t + \Delta t) - W(t)).
\end{equation*}

\noindent Therefore, $Y(t+ \Delta t) \sim N(\ln x + b_{\boldsymbol{S}} \Delta t, \sigma^2 \Delta t)$ and so its conditional expectation is 

\begin{equation*}
\mathbb{E}_{t,x} [ Y(t+\Delta t)] = \mathbb{E}_{t,x} [\ln N(t + \Delta t)] = \ln x + b_{\boldsymbol{S}} \Delta t.
\end{equation*}

\noindent Replacing into equation (\ref{eq:geometric1}), we obtain

\begin{align*} 
R_g(x,t) &= \frac{1}{x} \lim_{\Delta t \downarrow 0} \frac{\exp (\ln x + b_{\boldsymbol{S}} \Delta t) - x}{\Delta t}\\
        &= \frac{1}{x} \lim_{\Delta t \downarrow 0} \frac{xe^{b_{\boldsymbol{S}} \Delta t} - x}{\Delta t}\\
        &= \frac{1}{x} \lim_{\Delta t \downarrow 0} \frac{x(e^{b_{\boldsymbol{S}} \Delta t} - 1)}{\Delta t} \\
        &= b_{\boldsymbol{S}}.
\end{align*}
        
\noindent $N(t+ \Delta t) = \exp(Y (t+\Delta t))$ is lognormally distributed with parameters $\ln x + b_{\boldsymbol{S}} \Delta t$ and $\sigma^2 \Delta t$, hence its conditional expectation is

\begin{align*} \mathbb{E}_{t,x} [ N(t + \Delta t)] &= \mathbb{E}_{t,x} [\exp (Y (t + \Delta t)]\\
&= \mathbb{E}_{t,x} [\exp (\ln x + b_{\boldsymbol{S}} \Delta t + \sigma (W(t+\Delta t) - W(t)))] \\
&= \exp \left ( \ln x + b_{\boldsymbol{S}} \Delta t + \frac{1}{2} \sigma^2 \Delta t \right ) \\
&= x\exp \left ( \left (b_{\boldsymbol{S}} + \frac{1}{2} \sigma^2 \right ) \Delta t \right ),
\end{align*}

\noindent where $\mathbb{E}_{t,x}[e^{\sigma (W(t+\Delta t) - W(t))}] = e^{\frac{1}{2} \sigma^2 \Delta t}$, since $W(t+  \Delta t) - W(t) \sim N(0,\Delta t)$.

\noindent Replacing this into equation (\ref{eq:arithmetic}), we obtain

\begin{align*} R_a(x,t) &= \frac{1}{x} \lim_{\Delta t \downarrow 0} \frac{x\exp((b_{\boldsymbol{S}} + \frac{1}{2} \sigma^2) \Delta t) - x}{\Delta t} \\
&= \lim_{\Delta t \downarrow 0} \frac{\exp[(b_{\boldsymbol{S}} + \frac{1}{2} \sigma^2) \Delta t] - 1}{\Delta t} \\
&= b_{\boldsymbol{S}} + \frac{1}{2}\sigma^2 \Delta t.
\end{align*}

\noindent The conclusion is that

\begin{align*} R_g(x,t) &\equiv b_{\boldsymbol{S}},\\
              R_a(x,t) &\equiv b_{\boldsymbol{S}} + \frac{1}{2} \sigma^2. 
\end{align*}
               
\noindent Hence, when using Stratonovich calculus, the ``average" growth rate $b_{\boldsymbol{S}}$ is specified as the {\bf{geometric average}} growth rate.\\

\noindent {\bf 6.2.4.3 {Conclusion}} \\

\noindent The conclusion is that when using It\^o calculus, the ``average" growth rate, $b_{\boldsymbol{I}}$ is specified as the {\bf{arithmetic average}} growth rate, and when using Stratonovich calculus, the ``average" growth rate, $b_{\boldsymbol{S}}$ is specified as the {\bf{geometric average}} growth rate\\

\noindent This fact instructs us to replace the unspecified growth rate $b$ by the specified average it truly represents. It is only in this manner that the results acquire meaning. Since $b_{\boldsymbol{S}}$ is indeed the geometric average growth rate $R_g$, we can conclude that the solution of the Stratonovich SDE density-independent growth model is 

\begin{equation*} N(t) = N_0 \exp [R_g t + \sigma W(t)]. \end{equation*}

\noindent Since $b_{\boldsymbol{I}}$ is indeed the arithmetic average growth rate $R_a$, we can conclude that the solution of the It\^o SDE density-independent growth model is

\begin{equation*} N(t) = N_0 \exp \left [ \left (R_a - \frac{\sigma^2}{2} \right ) t + \sigma W(t) \right ], \end{equation*}

\noindent or, alternatively it was shown that $R_a - \frac{\sigma^2}{2} = R_g$, yielding

\begin{equation} N(t) = N_0 \exp[R_g t + \sigma W(t)]. \end{equation}

\noindent Hence, we can conclude that the two interpretations yield {\emph{exactly the same}} solutions in terms of a specific average growth rate. Therefore, it does not matter which average we choose as long as it is clearly specified. With regard to the conditions under which extinction occurs, we conclude that both approaches predict population extinction or a stochastic equilibrium according to whether the geometric average growth rate is negative or positive. So, we can use either calculus indifferently as long as we are careful to use $b$ for the appropriate average for that calculus.\\

\noindent Once $R_a(x,t)$ or $R_g(x,t) = R_a(x,t) - \frac{\sigma^2 (t)}{2}$ have been estimated, one can choose to estimate $b_{\boldsymbol{I}}(t) = R_a(x,t) = R_g(x,t) + \frac{\sigma^2 (t)}{2}$ and use It\^o calculus or choose to estimate $b_{\boldsymbol{S}}(t) = R_g(x,t) = R_a(x,t) - \frac{\sigma^2 (t)}{2}$ and use Stratonovich calculus. \\

\noindent It does not matter what choice one makes, the solution one obtains is the same.
Therefore, the It\^o and Stratonovich SDEs can be written in terms of these estimated quantities, $R_a(x,t)$ and $R_g(x,t)$:

\begin{align*}
\textnormal{It\^o SDE:} \qquad \frac{1}{N} \frac{dN}{dt} &= R_a(x,t) + \sigma (t) \xi (t) = R_g(x,t) + \frac{\sigma^2}{2} + \sigma (t) \xi (t) \\
\textnormal{Stratonovich SDE:} \qquad \frac{1}{N} \frac{dN}{dt} &= R_g(x,t) + \sigma (t) \xi (t) = R_a(x,t) - \frac{\sigma^2}{2} + \sigma (t) \xi (t)
\end{align*}

\noindent which are equivalent.

\subsection{Density-Dependent Growth}

\noindent We can interpret (\ref{eq:model}) as an It\^o SDE and it can be written as:

\begin{equation}
\label{eq:modelIto}
d N(t) = N(t) \Big[ b_{\boldsymbol{I}}(t)dt + \sigma(t) dW_t \Big].
\end{equation}

\noindent By the It\^o-Stratonovich conversion formula given in Equation (\ref{eq:convert1}), where $a(N_t,t) = b(t) N(t)$, $\sigma(N_t,t) = \sigma(t) N(t)$ and $\frac{\partial \sigma}{\partial x} (N_t,t) = \sigma(t)$. \\

Therefore the It\^o SDE is equivalent to the Stratonovich SDE

\begin{align}
\label{eq:modelStratcon}
dN(t) &= \left [ b_{\boldsymbol{I}}(t)N(t) - \frac{1}{2} \sigma^2 (t) N(t) \right ]dt + \sigma (t) N(t) dW_t, \\
dN(t)&= N(t) \left [ b_{\boldsymbol{I}}(t) - \frac{\sigma^2 (t)}{2} \right ]dt + \sigma (t) N(t) dW_t\\
&= b_{\boldsymbol{S}} N(t)dt + \sigma (t) N(t) dW_t.
\end{align}

\noindent This can be written in terms of the per capita average growth rate as
\begin{align*} 
\frac{1}{N} \frac{dN(t)}{dt} &= \left ( b_{\boldsymbol{I}}(t) - \frac{\sigma^2 (t)}{2} \right ) + \sigma (t) \xi (t)\\
&= b_{\boldsymbol{S}} + \sigma (t) \xi (t).
\end{align*}

\noindent This is similar to Equation (\ref{eq:model}) but simply interpreted as a Stratonovich SDE, where $b(t) \equiv b_{\boldsymbol{I}} (t)$ is replaced by $b(t) - \frac{\sigma^2 (t)}{2} \equiv b_{\boldsymbol{S}}$. Whether we interpret equation (\ref{eq:model}) as an It\^o or a Stratonovich SDE, the solution is a homogeneous diffusion process with diffusion coefficient (variance rate), $N^2\sigma^2$ (which is the same in both the It\^o and Stratonovich SDEs). The drift coefficient is, however, different; it is respectively for It\^o and Stratonovich:

\begin{align*}
\textnormal{It\^o}: \quad \mu(N) &= N(t)b_{\boldsymbol{I}}(t) \\
\textnormal{Stratonovich}: \quad \mu(N) &= N(t) \left [ b_{\boldsymbol{I}}(t) + \frac{\sigma^2 (t)}{2} \right ]\\
&= N(t)b_{\boldsymbol{S}} (t).
\end{align*}

\noindent To conclude, we have:

\begin{align*}
\noindent \textnormal{It\^o SDE}: \frac{1}{N(t)} \frac{dN(t)}{dt} &= b_{\boldsymbol{I}}(t) + \sigma (t) \xi (t) \\
\quad dN(t) &= N(t)b_{\boldsymbol{I}}(t)dt + \sigma(t)N(t)dW(t), \\
\noindent \textnormal{Stratonovich SDE}: \frac{1}{N(t)} \frac{d N(t)}{dt} &= b_{\boldsymbol{S}}(t) + \sigma(t) \xi (t)\\
\quad dN(t) &= N(t)b_{\boldsymbol{S}}(t)dt + \sigma(t)N(t)dW(t).
\end{align*}

\subsubsection{Deterministic Model}

\noindent We need to clarify what the growth rate and ``average" growth rate mean in terms of the observed population dynamics $N(t)$. Let us consider the deterministic model, where $\sigma = 0$, 

\begin{equation}
\label{eq:deter}
dN(t) = b_d (t) N(t) dt.
\end{equation}

\noindent We can define the growth rate (i.e. the per capita growth rate) at time $t$ when the observed population size at time $t$ is $x$, i.e. $N(t) = x$, as

\begin{align}
b_d(x) :&= \frac{1}{x} \frac{dN(t)}{dt} \\
&= \frac{1}{x} \lim_{\Delta t \downarrow 0} \frac{N(t+\Delta t) - N(t)}{\Delta t}\\
\label{eq:deter1} &= \frac{1}{x}\lim_{\Delta t \downarrow 0} \frac{N(t+\Delta t) - x}{\Delta t}.
\end{align}

\noindent The limit part represents the total growth rate, $\frac{dN(t)}{dt}$, at time $t$, which is then divided by the population size $x$ to obtain the per capita growth rate. In this density-independent case, $b_d(x)$ does not depend on the population size $x$, however, we adopt the notation $b_d (x)$ since the dependence will be present in the more general density-dependence case.\\

\noindent Alternatively, we could have obtained the solution of (\ref{eq:deter}) as $N(t) = N_0 \exp (b_d t)$, where this gives $N(t + \Delta t) = N(t) \exp (b_d \Delta t)$, and substituting into equation (\ref{eq:deter1}) to obtain 

\begin{equation*}
\frac{1}{x} \lim_ {\Delta t \rightarrow 0} \frac{ x (e^{b_d \Delta t} - 1)}{\Delta t}.
\end{equation*} 

\noindent Furthermore, we know that as $\Delta t \downarrow 0$, $ \frac{e^{b_d \Delta t} - 1}{\Delta t}$ tends to $b_d$, to arrive at the same solution.\\

\noindent However, in the stochastic model (\ref{eq:model}), i.e. $\sigma \neq 0$, $N(t + \Delta t)$ is a random variable and it is then necessary to take some kind of average of $N(t + \Delta t)$, to obtain a possible estimate of the population size at time $t + \Delta t$. One possible way is to take the limit and determine the average afterwards, however, this is shown not to work since the limit itself is a generalised stochastic process and does not exist in the ordinary sense. So, instead, we follow the other approach by first taking the average and then computing the limit afterwards. But we have to be precise on what type of average we are using.

\subsubsection{It\^o Model}

\noindent Let us see what these two averages turn out to be, under the It\^o model. It is sometimes more convenient to work with the log transformed stochastic process $Y(t) = \ln N(t)$ and with $y = \ln (x)$. Using It\^o's rule of calculus, one obtains, using the fact that in the limit as $dt$ tends to 0 ($dt \rightarrow 0$), $(dt)^2 = 0$, $(dt)(dW_t) = 0$ and $(dW_t)^2 = dt$:

\begin{equation} 
dN(t) = b_{\boldsymbol{I}}(t) N(t) dt + \sigma (t) N(t) dW (t).
\end{equation}

\noindent This It\^o model has drift coefficient $b_{\boldsymbol{I}}(t)N(t)$ and diffusion coefficient $\sigma^2 (t) N^2 (t)$, where the growth rate $b_{\boldsymbol{I}}(t)$ is given by;

\begin{equation*}
b_{\boldsymbol{I}}(t) = \frac{1}{x} \lim_{\Delta t \downarrow 0}\frac{\mathbb{E}_{t,x}[\ N(t+\Delta t)]-x}{\Delta t} = R_a(x,t).
\end{equation*}

\noindent The arithmetic average growth rate is then given by $R_a(x,t) \equiv b_{\boldsymbol{I}}(t)$. \\

\noindent To compute the geometric average growth rate, let us consider the log scale $Y(t) = \ln N(t)$, then using It\^o's lemma, we obtain the SDE (\ref{eq:log}):

\begin{align} dY = d\ln N(t) &= \frac{d\ln N(t)}{dN(t)} dN(t) + \frac{1}{2} \frac{d^2\ln N(t)}{dN(t)^2} (dN(t))^2\\
                 &= \frac{1}{N(t)} \left [ N(t)b_{\boldsymbol{I}}(t)dt + \sigma(t)N(t)dW(t) \right ] + \frac{1}{2} \left ( \frac{-1}{N(t)^2} \right ) \left [ \sigma^2(t)N^2 (t) dt \right ]\\
 \label{eq:log}                &= \left [ b_{\boldsymbol{I}}(t) - \frac{1}{2} \sigma^2 (t) \right ] dt + \sigma(t)dW(t).
                 \end{align}

\noindent In terms of $Y(t) = \ln N(t)$, $y = \ln (x)$, the solution of equation (\ref{eq:log}) has drift coefficient $b_{\boldsymbol{I}}(t) - \frac{1}{2} \sigma^2 (t)$ and diffusion coefficient $\sigma^2 (t)$, where the growth rate $b_{\boldsymbol{I}}(t) - \frac{1}{2} \sigma^2 (t)$ is given by:

\begin{equation*}
b_{\boldsymbol{I}}(t) - \frac{1}{2} \sigma^2 (t) = \lim_{\Delta t \downarrow 0}\frac{\mathbb{E}_{t,x}[\ Y(t+\Delta t)]-y}{\Delta t} = R_g(x,t).
\end{equation*}

\noindent Therefore, the geometric average growth rate is given by

\begin{equation*}
R_g(x,t) = \lim_{\Delta t \downarrow 0}\frac{\mathbb{E}_{t,x}[\ \ln N(t+\Delta t)]-\ln x}{\Delta t} = b_{\boldsymbol{I}}(t) - \frac{1}{2} \sigma ^2 (t) = b_{\boldsymbol{S}} (t).
\end{equation*}

\noindent Hence, for the It\^o SDE, we have that
\begin{align*} R_a(x,t) &= b_{\boldsymbol{I}}(t),\\
               R_g(x,t) &= b_{\boldsymbol{I}}(t) - \frac{\sigma^2}{2} = b_{\boldsymbol{S}} (t).
               \end{align*}

\noindent These are respectively the arithmetic average growth rate (the expected value of average {\emph{w.r.t}} the process $N(t)$ as per the definition of the average (\ref{eq:arithmetic})) and the geometric average growth rate defined by (\ref{eq:geometric1}) for the solution of the It\^o SDE.
              
\noindent By the definition (\ref{eq:geometric1}) of the geometric average, $R_g(x,t)$ is the average {\emph{w.r.t}} the process $\ln N(t)$, so by determining the dynamics of $\ln N(t)$ we can obtain the drift rate which gives the per capita geometric growth average. 

\noindent To conclude $b_{\boldsymbol{I}}(x,t)$ is the {\bf{arithmetic average}} growth rate
\begin{equation*} 
R_a(x,t) := \frac{1}{x} \lim_{\Delta t \downarrow 0}\frac{\mathbb{E}_{t,x}[\ N(t+\Delta t)]-x}{\Delta t};
\end{equation*}

\noindent and $b_{\boldsymbol{I}}(t) - \frac{\sigma^2 (t)}{2} = b_{\boldsymbol{S}}$ is the {\bf{geometric average}} growth rate
\begin{equation*} 
R_g(x,t) := \lim_{\Delta t \downarrow 0}\frac{\mathbb{E}_{t,x}[\ \ln N(t+\Delta t)]-\ln x}{\Delta t}.
\end{equation*}

\noindent Therefore, $R_g(x,t) = R_a(x,t) - \frac{\sigma^2 (t)}{2}$.

\subsubsection{Stratonovich Model}

\noindent For the Stratonovich model, we make use of an easier approach which is to convert the Stratonovich SDE to an equivalent It\^o SDE using the It\^o-Stratonovich conversion formula (\ref{eq:convert2}):
\begin{align*}
\textnormal{Stratonovich SDE}: \frac{1}{N} \frac{dN(t)}{dt} &= b_{\boldsymbol{S}}(t) + \sigma(t) \xi (t), \\ 
\label{eq:Strateq} dN(t) &= N(t)b_{\boldsymbol{S}}(t)dt + \sigma(t) N(t)dW(t)\\
\textnormal{equivalent It\^o SDE}: \frac{1}{N(t)} \frac{dN (t)}{dt} &= b_{\boldsymbol{S}}(t) + \frac{\sigma^2 (t)}{2} + \sigma (t) \xi(t) \\
dN(t) &= N(t)b_{\boldsymbol{S}}(t)dt + N(t) \left (\frac{\sigma^2 (t)}{2} \right ) dt + \sigma(t) N(t) dW(t), \\
&= N(t) \left [ b_{\boldsymbol{S}}(t) + \frac{\sigma^2(t)}{2} \right ] dt + \sigma (t) N(t) dW(t).
\end{align*}

\noindent The solution of the Stratonovich SDE is a diffusion process with drift coefficient $ \left (b_{\boldsymbol{S}}(t) + \frac{\sigma^2 (t)}{2} \right ) N(t)$ and diffusion coefficient $\sigma^2 (t) N^2 (t)$, which is identical to the diffusion coefficient of the It\^o SDE.\\

\noindent If we now consider the transformation $Y(t) = \ln N(t)$, which is instrumental in these deductions. Using It\^o's Lemma  we obtain
\begin{align} dY(t) &= \frac{d\ln N(t)}{dN(t)} dN(t) + \frac{1}{2} \frac{d^2\ln N(t)}{dN(t)^2} (dN(t))^2,\\
&= \frac{1}{N(t)} \left [ N(t) \left ( b_{\boldsymbol{S}}(t) + \frac{\sigma^2 (t)}{2} \right )dt + \sigma(t)N(t)dW(t) \right ] + \frac{1}{2} \left ( \frac{-1}{N(t)^2} \right ) \left [ \sigma^2(t)N^2 (t) dt \right ],\\
\label{eq:Strateq1} &= b_{\boldsymbol{S}}(t)dt + \sigma(t)dW(t).
                 \end{align}
                 
\noindent The solution of equation (\ref{eq:Strateq1}) has drift $b_{\boldsymbol{S}}(t)$ and diffusion coefficient $\sigma^2 (t)$.\\

\noindent Alternatively, we could compute the geometric average growth rate by starting from the Stratonovich SDE and using the ordinary chain rule of differentiation, we obtain the SDE:

\begin{align*}
d \ln N(t) &= \frac{d \ln N(t)}{d N(t)} dN(t) \\
           &= \frac{1}{N(t)} \left [ N(t) b_{\boldsymbol{S}} (t) dt + \sigma (t) N(t) dW(t) \right ] \\
           &= b_{\boldsymbol{S}}(t)dt + \sigma (t) dW (t).
\end{align*}  

\noindent This can indifferently be interpreted as an It\^o or a Stratonovich SDE. Since, the stochastic term has a constant coefficient, the correction term in the conversion method is now zero and the two approaches coincide.
                 
\noindent Therefore, using a similar reasoning as explained with the It\^o SDE, for the Stratonovich SDE, we have 

\begin{align*} R_g(x,t) &= b_{\boldsymbol{S}}(t),\\
                        &= b_{\boldsymbol{I}}(t) - \frac{1}{2} \sigma^2\\
               R_a(x,t) &= b_{\boldsymbol{S}}(t) + \frac{\sigma^2}{2}\\
                        &= b_{\boldsymbol{I}}(t).
               \end{align*}

\noindent These are respectively the geometric average growth rate and the arithmetic average growth rate for the solution of the Stratonovich SDE. This proves that the average used in Stratonovich calculus is the geometric average and the average used in It\^o calculus is the arithmetic average. Hence, from these results, we reach the following final conclusions:

\noindent To conclude the arithmetic average growth rate is given by

\begin{equation*} 
R_a(x,t) := \frac{1}{x} \lim_{\Delta t \downarrow 0}\frac{\mathbb{E}_{t,x}[\ N(t+\Delta t)]-x}{\Delta t} = b_{\boldsymbol{S}}(t) + \frac{\sigma^2}{2},
\end{equation*}

\noindent and the geometric average growth rate is given by
\begin{equation*}
R_g(x,t) := \lim_{\Delta t \downarrow 0}\frac{\mathbb{E}_{t,x}[\ \ln N(t+\Delta t)]-\ln x}{\Delta t} = b_{\boldsymbol{S}}(t).
\end{equation*}

\noindent Again, for the Stratonovich SDE, we also have that $R_g(x,t) = R_a(x,t) - \frac{\sigma^2}{2}$.

\section{Conclusion}

\noindent Under {\bf{It\^o calculus}}, we interpret the growth rate as the {\bf{arithmetic average}} growth rate, $R_a(x,t)$ defined by equation (\ref{eq:arithmetic}).\\

\noindent Under {\bf{Stratonovich calculus}}, we interpret the growth rate as the {\bf{geometric average}} growth rate, $R_g(x,t)$ defined by equation (\ref{eq:geometric1}).

\noindent Therefore, for It\^o calculus, $b(t)$ really means the arithmetic average growth rate, $R_a(x,t)$ or equivalently $R_g(x,t) + \frac{\sigma^2}{2}$ and for Stratonovich calculus, $b(t)$ really means the geometric average growth rate, $R_g(x,t)$. \\

\noindent It is shown, and finally concluded, in Braumann (2007) that both calculus lead to the exact same conclusions in terms of the conditions under which population extinction or the existence of a stochastic equilibrium occur. Since if one takes into account the difference $\frac{\sigma^2}{2}$ between the two averages, the solutions under both approaches coincide. \\

\noindent Hence, after clearing the confusion, It\^o and Stratonovich calculus yield the same results. It is now possible to easily tackle the major obstacle to the use of these SDE models, however, care must be taken in using the appropriate type of average for each calculus If that care is indeed taken, It\^o and Stratonovich will give the same results and draw the same conclusions.


\begin{thebibliography}{99}

\bibitem{B2007} Braumann  C. A., (2007). \emph{Harvesting in a random environment: It\^o or Stratonovich calculus?} J. Theoret. Biol. 244, no. 3, 424--432.
  
\bibitem{tnss} Fernholz, R., and  B. Shay, (1982). \emph{Stochastic portfolio theory and stock market equilibrium} Journal of Finance 37, 615--624.


\end{thebibliography}
\end{document}